\documentclass[12pt]{article}

\usepackage{amsmath,amsthm,amssymb,amsxtra, graphicx}
\usepackage{color}

\usepackage[latin1]{inputenc}
\usepackage[english]{babel}
\usepackage{amsmath,amstext,amsbsy,amsfonts,eucal,latexsym,amssymb}
\usepackage{indentfirst}
\usepackage{graphicx}
\usepackage{setspace}
\usepackage{float}
\usepackage{placeins}
\usepackage{xparse}
\usepackage{tikz}
\usepackage{braket}
\usepackage{longtable}
\usepackage{tabularx,ragged2e,booktabs}
\usepackage{multicol}

\newcommand{\C}{{\mathbb{C}}}

\newcommand{\demo}{{\textit{\bf Proof.}}}

\newcommand{\fimh}{\hfill{$\square$}\vspace{0.35cm}}

\newtheorem{teo}{Theorem}[section]
\newtheorem{prop}[teo]{Proposition}

\newtheorem{df}{Definition}
\usepackage{geometry}
\usepackage{epsfig}
\usepackage{multicol}
\begin{document}

	\title{
		Asymmetric quantum codes on non-orientable surfaces%
	}
	
	\author{
		Waldir S. Soares Jr., Douglas F. Copatti, Giuliano G. La Guardia \\
		and Eduardo B. Silva%
		\thanks{ebsilva@uem.br}}

\date{ }	

\maketitle	
	
\begin{abstract}
In this paper, we construct new families of asymmetric quantum surface
codes (AQSCs) over non-orientable surfaces of genus $g\geq 2$ by
applying tools of hyperbolic geometry. More precisely, we prove that if the genus
$g$ of a non-orientable surface is even $(g=2h)$, then
the parameters of the corresponding AQSC are equal to the parameters
of a surface code obtained from an orientable surface of genus $h$. Additionally,
if $S$ is a non-orientable surface of genus $g$, we show that the new
surface code constructed on a $\{p, q\}$ tessellation over $S$ has
the ratio $k/n$ better than the ratio of an AQSC constructed on the same
$\{p, q\}$ tessellation over an orientable surface of the same genus $g$.

\vspace{0.2cm}		

		
\end{abstract}

\section{Introduction}
	
One of the most important class of codes, which includes the Calderbank-Shor-Steane (CSS)
codes \cite{Calderbank:1998}, is the class of stabilizer codes, introduced by Gottesman
\cite{gottesmanstabilizer}. These codes play similar rules of the linear codes in classical
coding theory. In this stabilizer formalism, the construction of quantum codes is performed
by choosing abelian subgroups of the Pauli group that leave invariant the code space.

Among all stabilizer codes, some topological quantum codes stand out. Topological quantum codes
encode quantum words in the nonlocal degrees of freedom of topologically ordered
physical systems. Because of this, these quantum codewords
are intrinsically resistant to effects of noise, as long as it remains local.

The well-known toric codes form an important class of topological quantum
codes, which was introduced by Kitaev \cite{Kitaev}. Since then, such codes
have been widely explored, extended e generalized in many different ways,
\cite{Albuquerque2009}, \cite{Bombin2006Topological},
\cite{Kitaev2}, \cite{delfosse}.

Surface codes are a powerful and promising method for protecting quantum
information from errors. By encoding quantum information into
a two-dimensional tessellation of qubits and using syndrome measurement as
well as logical operations to detect and correct errors, surface codes
offer a way to build large-scale, fault-tolerant quantum computers that
can perform complex computations and simulations. Surface codes are relatively
robust to certain types of errors, such as errors that occur in
clusters or patterns.

Asymmetric quantum error-correcting codes (AQECCs) form a class of quantum
codes that are capable of correcting different types of errors, i.e.,
phase-shift and qubit-flip errors. The combined amplitude damping and dephasing
channel \cite{Sarvepalli:2008} is an example of
a quantum channel that satisfies $d_{z} > d_{x}$, that is, the probability
of occurrence of phase-shift errors is greater than the probability of occurrence
of qubit-flip errors. The idea of correcting asymmetric errors with unequal
probabilities was introduced by Steane \cite{Steane1996Simple}. After such a work,
much research on this type of codes has already been carried out:
detailing construction methods, structures and performance, as can be seen
in \cite{Evans2007}, \cite{Ezerman2011},
\cite{Ioffe2007}, \cite{Guardia2011}, \cite{Guardia2012}, \cite{Guardia2016}.

In an AQECC, the encoding operation is optimized for correcting errors that
occur during transmission, while the decoding operation is optimized
for correcting errors that occur during storage. This approach is
particularly useful in scenarios where the probability of errors occur
during transmission is different from that during storage. For example,
in a quantum communication scenario, errors may be more likely to occur
during transmission due to the presence of noise in the communication channel.

Asymmetric quantum codes can be designed using a variety of techniques,
such as stabilizer and subsystem codes. In case of stabilizer codes,
the encoding and decoding operations are based on the stabilizer formalism,
which involves the utilization of stabilizer generators to
detect and correct errors. In subsystem codes, the encoding and
decoding operations are based on the theory of quantum subsystems,
which involves the division of the system into
smaller subsystems in which the errors can be individually corrected.

The main contribution of this paper is to construct new families of asymmetric quantum
surface codes (AQSCs) over non-orientable surfaces of genus
$g\geq 2$, utilizing tools of hyperbolic geometry. We show that if the genus $g$ of a
non-orientable surface is even $(g=2h)$, then the parameters of the corresponding
AQSC are equal to the parameters of an AQSCs over an orientable surface of
genus $h$ (see Theorem~\ref{teo1}). Furthermore, we also show that if $S$
is a non-orientable surface of genus $g$, then the new surface code constructed
on a $\{p, q\}$ tessellation over $S$ has
ratio $k/n$ better than the ratio of an AQSC derived from the same tessellation over
an orientable surface of the same genus $g$ (see Theorem~\ref{mainteo}).
Asymmetric surfaces codes over orientable
surfaces was presented in \cite{Albuquerque2022}.

This paper is organized as follows. Section~\ref{hyperbolic geometry} presents
a brief review on hyperbolic geometry and the
connection between hyperbolic polygons and two-dimensional surfaces 
(orientable or non-orientable).
In Section~\ref{sec:surface}, we review
the concept of surface codes which, together with hyperbolic geometry tools, provides a
technique to generate new families of asymmetric surface
codes. In Section~\ref{newcodes}, we present the main contribution of this paper: the construction
of new families of AQSCs over non-orientable surfaces. Finally, in Section~\ref{final},
the final remarks are drawn.

\section{Hyperbolic geometry} \label{hyperbolic geometry}

In this section we present definitions and results of hyperbolic geometry that will be necessary
for the development of the article. For more details on hyperbolic geometry, the reader can consult
\cite{Stillwell1995, Beardon2012, Katok1992}.

As a model of the hyperbolic plane, we consider the \textit{upper-half plane}
$\mathbb{H}^2=\{z\in \mathbb{C}: Im(z)>0\}$. A convex closed set formed
by $p'$ geodesic segments is a \emph{hyperbolic polygon} (or a $p'$-gon) $\varPi'$ with edges
$p'$. A $p'$-gon is \textit{regular} if its edges have the same length and the interior
angles are equal. We denote by $\mathbb{S}^2$ the sphere in $\mathbb{R}^3$ and by $\mathbb{E}^2$
the Euclidean plane.

The \emph{unimodular group} $SL(2,\mathbb{R})$ (or $SL_2(\mathbb{R})$) is the group of $2 \times 2$
real matrices with determinant one:

\[SL_2(\mathbb{R})=  \left\{
\begin{pmatrix}
	a & b\\
	c & d
\end{pmatrix}  \biggm| a, b, c, d \in \mathbb{R}, ad-bc =1\} \right \};
\]
and the \emph{linear special projective group} ${\operatorname{PSL}}_2(\mathbb{R})$
is the quotient

\[ PSL_2(\mathbb{R}) =SL_2(\mathbb{R})/\{\pm I_2\}. \]

The elements of $PSL_2(\mathbb{R})$ act on the complex plane by \emph{M\"{o}bius transformations}
$f:\mathbb{C}\cup\{\infty\}\rightarrow\mathbb{C}\cup\{\infty\} $, where
$f(z)=\frac{az+b}{cz+d} $, $ a, b, c, d\in \mathbb{R}$ and $a d - bc = 1$.  This is precisely the set
of M\"{o}bius transformations that preserve the upper half-plane and they are its isometries.
The group $PSL_2(\mathbb{R})$ is the group of conformal automorphisms of the upper
half-plane (see \cite{Katok1992} for more details). From the Riemann mapping theorem,
it is also isomorphic to the group of conformal automorphisms of the unit disc. Thus,
the hyperbolic area of a hyperbolic polygon is also invariant under images of the members of $PSL_2(\mathbb{R})$.
A \emph{Fuchsian group} is a discrete subgroup of $ PSL_2(\mathbb{R})$.

It is well-known that every compact and connected $2$-manifold is
homeomorphic either to $\mathbb{S}^2$, or to a connect sum of one or more
copies of the torus, or a connect sum of one or more copies of the projective
plane. It is also known that these admit polygonal presentations
(see \cite{Stillwell1995}, for instance).

The area of a hyperbolic triangle can be given only by its angles. Moreover, the area
of a hyperbolic triangle and the area
of any regular hyperbolic polygon can be computed from the well-known
Gauss-Bonnet Theorem.

\begin{teo}[Gauss-Bonnet]\label{Gauss-Bonnet}
If $T$ is a hyperbolic triangle with internal angles $\alpha, \beta$ and $\gamma$,
then the area of $T$ is given by
	\begin{eqnarray}
		area(T)=\pi-\alpha-\beta-\gamma \,.
	\end{eqnarray}
\end{teo}

A \emph{regular tessellation} of Euclidean or hyperbolic plane is a covering of the plane by
copies of a regular polygon, where the degree of each vertex is the same, such that the
intersection of two faces is either empty, or along entire edges or at vertices, and the interior
of two distinct faces are disjoint. The degree of a vertex is the number of edges meeting
at the vertex. A regular tessellation is denoted by $\{p, q\}$ if $q$ regular $p$-gons meet
at each vertex. We write $\varPi'$ to indicate a regular hyperbolic polygon associated with
a fundamental region of the $ \{p', q'\} $ tessellation and, similarly, $\varPi$ for the
$\{p, q\}$ tessellation.

It can be proved that any compact topological surface $M$ can be realized geometrically
from a polygon $\varPi'$ by \emph{edge-pairing identification}, see \cite{Stillwell1995}.
An edge-pairing identification is an element $\zeta (\neq Id)$ of an isometry group $ \Gamma $
that preserves the orientation, such that $\zeta$ identifies an edge $e$ of $\varPi'$
to another edge $e'$ of $\varPi'$ (i.e., $ \zeta(e)=e' $). Moreover,
$\zeta^{-1}\in\Gamma\setminus Id $ and $ \zeta^{-1}(e')=e $.
In general, if $ \zeta_1(e)=e'$ and $ \zeta_2(e')=e'' $ for some $ \zeta_1,\zeta_2\in\Gamma$,
then $ \zeta(e)=e'' $, where $ \zeta=\zeta_2\circ\zeta_1 $.
Hence, if $\Gamma$ is the Fuchsian group generated by the edge-pairing identifications
$\zeta$ of $\varPi'$, it gives rise to an \textit{identification space} $ S_{\varPi'} $.
This is a hyperbolic surface since $\Gamma$ is a Fuchsian group and $\varPi'$ is a
fundamental region of the $\{p',q'\}$ tessellation. Thus, $ S_\varPi' = \varPi'/\Gamma$
and $M =S_\varPi'$. Recall that given a space and a group acting on it, the
images of a single point under the group action form an orbit of the action. A
\textit{fundamental ~region} \cite{Katok1992} is a subset of the space that contains
exactly one point from each of these orbits. It serves as a geometric realization
for the abstract set of representatives of the orbits. In the context above, an
identification may also exist with vertices. We indicate a maximal set
$ \{v_1,v_2,\dots, v_k\} $ of identified vertices as \textit{vertex cycle}.

We next present some results from \cite{Stillwell1995} which will be utilized in this work.
Propositions \ref{edge_angle_condition} and \ref{poincare}
establish the conditions for a polygon to be a fundamental region.

\begin{prop}[Edge and Angle condition]\label{edge_angle_condition}
Let $\varPi'$ be a compact polygon and let $\Gamma$ be an isometry
group of $\mathbb{S}^2$, $ \mathbb{E}^2 $ or $ \mathbb{H}^2 $ which
preserves orientation, such that $\varPi'$ is a fundamental region for $\Gamma$.
Then one has:
	
\begin{enumerate}
		\item For each edge $s$ in $\varPi'$ there exists exactly one edge $s'(\neq s)$
of $ \varPi $ such that $ s'=\gamma(s) $, $ \gamma\in \Gamma $ (the elements of
$\Gamma $ are called side pairing transformations of $ \varPi' $).
		\item For each set of paired edges of $ \varPi' $, the sum of angles of the
set of identified vertices is $ 2\pi $. This set is a \emph{vertex cycle}.
\end{enumerate}
\end{prop}

\begin{prop}[Poincar\'e]\label{poincare}
If a compact polygon $\varPi'$ satisfies the Edge and Angle conditions,
then $\varPi' $ is a fundamental region for the group $ \Gamma $ generated by
the edge-pairing transformations of $ \varPi' $.
\end{prop}

We know that every non-orientable surface of genus $g$ can be presented by a $2g$-gon
$ \varPi' $. We label the edges of the polygon $ \varPi' $ that must be identified by
the same symbols, and the identification should be made according to the orientation given
in the edge. We label edges with a symbol, say ``$ a $" for a particular orientation,
and ``$ a^{-1} $" for the opposite orientation. Therefore, we can get a surface of genus
$g$ from a $2g$-gon $ \varPi' $, by labelling the edges as $ a_1, a_1 , a_2 ,  a_2 ,
\ldots, a_g , a_g $. We write this polygon, $ \varPi' $ as $ a_1a_1a_2a_2\dots a_g a_g $.
The identification space $ a_1 a_1 a_2 a_2\dots a_g a_g $ of a $ 2g $-gon can be also
generated by diagonal edge pairing identification of the $2g $-gon by applying cut
and paste processes (see \cite{Girondo2012}).
For instance, the cut and paste process for a normalized edge pairing identification
to get the diagonal edge pairing identification, is illustrated in Figure~\ref{cut and paste}
for a non-orientable surface of genus $3$.

\begin{figure}[ht]
	\includegraphics[scale=0.8]{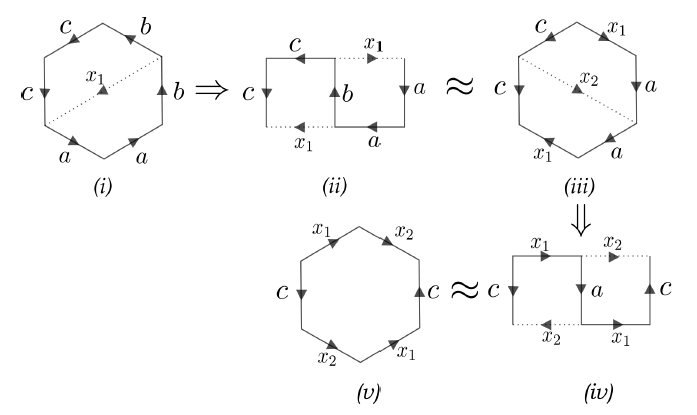}
	\vspace*{10pt}
	\caption{\label{cut and paste}Cut and paste process where in $(i) $ is the given polygon. Cutting $ (i) $ diagonally along $ x_1 $ and pasting along $ b $ gives $ (ii) $, which is isomorphic to $ (iii) $. Cutting $ (iii) $ diagonally along $ x_2 $ and pasting along $ a $ gives $ (iv) $, which is isomorphic to $ (v) $.}
\end{figure}

The method of constructing topological quantum codes utilizes a $2g$-gon $\varPi'$
of $\mathbb{H}^2$. Let $ S=\{s_1,s_2, \ldots, s_{2g}\} $ be the set of edges
of $ \varPi' $. The edge pairing transformation for the polygon $\zeta: S\rightarrow S$
is given by $ \zeta(s_i)=s_{g+i} $, where the sum of the subscripts of $s$ is
performed modulo $2g$. In $S$, $ s_{i} $ is adjacent with $ s_{i-1} $ and $ s_{i+1} $
for all $ i\in\{ 2,3,\dots, 2g-1\}$, $ s_{1}$ is adjacent with $ s_{2g} $ and $ s_{2} $,
and $ s_{2g} $ is adjacent with $s_{2g-1}$ and $ s_{1} $. These identifications produce
a code having maximal hyperbolic distance between the identified edges of $ \varPi' $
with the edge-pairing. For instance, if $ S=\{s_1,s_2, \dots, s_{10}\} $ then
$ \zeta: S\rightarrow S $ satisfy $\zeta(s_1) = s_{6},~ \zeta(s_2) = s_{7},
~\zeta(s_3) = s_{8},~ \zeta(s_4) = s_{9},~ \zeta(s_5) = s_{10}$
(see Fig.~\ref{edge pairing identification}).

\begin{figure}[ht]
	\vspace*{10pt}
		\begin{center}
		\includegraphics[scale=0.6]{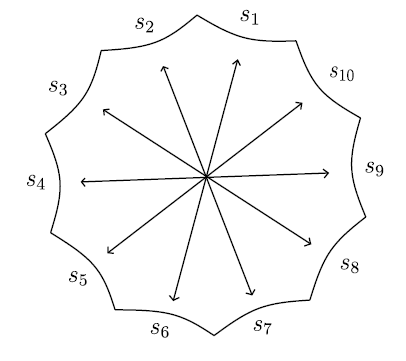}
		
	\end{center}
	\caption{Edge-pairing transformation $ \gamma(s_i)=s_{g+i} $}
	\label{edge pairing identification}
\end{figure}

\emph{Euler characteristic} is an important topological invariant for surfaces. Given
a compact surface $M$ that can be tessellated by a finite number of polygons,
let $V$, $E$ and $F$ be the set of vertices, edges, and faces, respectively, determined
by a tessellation of the surface. The Euler characteristic $\chi(M)$ of $M$
is defined by	
	
\begin{eqnarray}
	\chi(M)=\#V-\#E+\#F \nonumber \,,
\end{eqnarray}
where $ \#V, \#E, \#F $ are the cardinalities of $V, E, F$, respectively.
$\chi(M)$ depends only on the topology of the surface and not on
a particular tessellation. For closed orientable surfaces of genus $h$,
it follows that (see \cite{Stillwell1995})
\begin{eqnarray}
	\chi(M)=2-2h \nonumber \,;
\end{eqnarray}
for closed non-orientable surfaces of genus $g$, one has
\begin{eqnarray}
	\chi(M) = 2-g \nonumber \,.
\end{eqnarray}

We can obtain a relation between the number of elements in a regular tessellation
in a surface. Let $\varPi'$ be a regular hyperbolic polygon associated with the
fundamental region of the $\{p', q'\}$ tessellation, that is, $\varPi'$ is a polygon with
$p'$ edges where $q'$ polygons with $p'$ edges meet in each vertex. If we count
the $q$ edges in each of the $V$ vertices, we have counted each edge of the
tessellation twice. Analogously, if we count all the $p$ edges corresponding
to the border of each of the $F$ faces of the tessellation, we have counted
each edge of the tessellation twice. Therefore, we have $qV = 2E = pF$.
	
\section{Surface codes} \label{sec:surface}
	
A \emph{quantum error-correcting code} (QECC) $Q$ is a $K$-dimensional subspace
of the complex Hilbert space ${\mathbb C}^{q^n}$. If $Q$ has minimum distance
$d$, then $Q$ is an $((n, K, d))_q$ code. If $K = q^k$, we write $[[n, k, d]]_q$. The
length $n$, the dimension $K$ and the minimum distance $d$ are the parameters of $Q$.

In this paper we only deal with $q=2$.
	
In~\cite{Kitaev}, Kitaev proposed a particular case of stabilizer code,
the well-known toric code. Utilizing a $l \times l$ square
as the fundamental polygon $P$ for a torus, Kitaev considered a $\{4, 4\}$ tessellation of $P$.
He associated a qubit to each edge of the corresponding tessellation. At each
vertex $v$ of the tessellation is associated an operator $X_v$ acting as
tensor products of the Pauli matrix $X$ on each edge adjacent to this vertex and
as identity operator in all other edges: \[ X_f = \bigotimes_{i \in E_f} X^i\].
Each face $f$ of the tessellation is associated
with an operator $Z_f$ which acts as the tensor product of Pauli matrix $Z$ on each
edge of the boundary of $f$ and as identity in all other faces:
\[Z_v = \bigotimes_{j \in E_v} Z^j\].

\begin{figure}[!h]
\centering
		\includegraphics{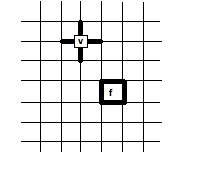}
		\caption{Face (f) e vertice (v) operators of Kitaev's Toric Code act on the highlighted qubits}
		\label{suportekitaev}
	\end{figure}
	
In this way, Kitaev's toric code is then defined by the space stabilized by $X_v$'s and $Z_f$'s:
	
	\[
	\mathcal{C}=\{|\psi \rangle; X_v|\psi\rangle=|\psi\rangle; Z_f|\psi\rangle=|\psi\rangle; \forall v, f \} \,.
	\]
	
The code length is the number of edges of the $\{4, 4\}$ tessellation,
i.e., $n=2l^2$.  The encoded qubits $k$ is the number of representatives of the equivalence classes
of homological non-trivial cycles; hence $k=2$, since the torus has genus $g=1$.
The minimum distance $d$ of a toric code is the minimum among the number of edges
contained in the shortest homological non-trivial cycle in the original
tessellation ($d_x$) and the number of edges contained in the shortest homological
non-trivial cycle in its dual tessellation ($d_z$), that is, $d = \min(d_x, d_z) = l$
(in this particular case, since the $\{4, 4\}$ tessellation is autodual, one has $d_x = d_z$).
	
There exist many types of extensions and generalizations of the toric code
\cite{Albuquerque2009}, \cite{Silva2021}). In \cite{Albuquerque2009}, the
authors worked with codes over orientable surfaces whereas in \cite{Silva2021}, the construction
was performed over non-orientable surfaces.

Applying similar ideas of Kitaev's toric code in order to generate a code,
it is necessary to tile the fundamental polygon $P^\prime$ utilizing
a $\{p, q\}$ tessellation. Let us consider a $g$-torus with fundamental polygon
$P^{\prime}= \{4g, 4g\}$, tessellated by a polygon $P= \{p, q\}$.
To make this tessellation possible,
the quotient between the area $P^{\prime}$ and $P$ must be a positive integer. If
$\mu(M)$ is the area of a polygon $M$, the number of faces $n_f$ can be computed by
\begin{equation}
    n_f = \frac{\mu({P^{\prime}})}{\mu({P})} = \frac{4q(g-1)}{pq-2p-2q}.
\end{equation}

The code length is given by
\begin{equation} \label{parn}
    n=\frac{p}{2} n_f.
\end{equation}

As was said previously, the encoded qubits $k$ is the number of homological
non-trivial cycles belonging to distinct equivalence classes of homology,
which implies $k=2g$. In order to compute the minimum
distance of the code, one must find out the hyperbolic distance $d_h$ between
two opposite sides of the fundamental
polygon $P^{\prime}$. Thus, a homologically non-trivial cycle has, at least, a
number of edges greater than $d_h$.

\begin{figure}[H]
\centering
		\includegraphics[scale=0.6]{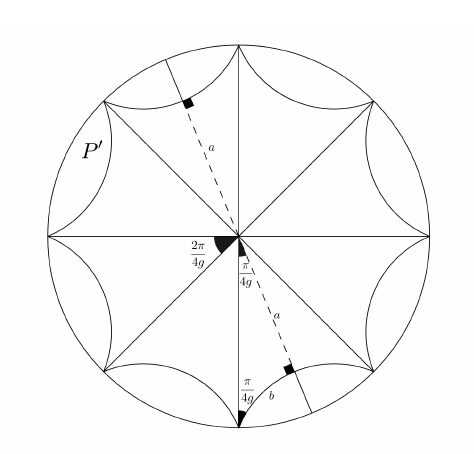}
		\caption{The hiperbolic distance between two opposite sides of a polygon $P^{\prime}$ is $d_h = 2a$}
		\label{figdh}
	\end{figure}

From hyperbolic trigonometry (see Fig.~(\ref{figdh}), it follows that
\begin{equation}
    d_h = 2arccos\frac{\cos(\pi/4g)}{\sin(\pi/4g)},
\end{equation}

and the length $l(p, q)$ of an edge of a $\{p, q\}$ polygon is given by
\begin{equation}
    l(p,q) = arccos\Bigg[\frac{\cos^2(\pi/q) + \cos(2\pi/p)}{\sin^2(\pi/q)}\Bigg].
\end{equation}

Hence,
\begin{equation}
    d_x = \frac{d_h}{l(p,q)} \hbox{ and } d_z = \frac{d_h}{l(q,p)},
\end{equation}

which implies

\begin{equation} \label{dist}
    d= \lceil{\min(d_x, d_z)}\rceil.
\end{equation}

\section{Asymmetric surface codes}\label{newcodes}
	
In this section, new families of AQSCs obtained from a
non-orientable surface are constructed. As it is usual,
we assume that $d_X$ is utilized to correct $X$ errors and $d_z$
to correct $Z$ errors.
	
\begin{df}\cite{Guardia2020} An asymmetric quantum error correction code
(AQECC) with parameters
$((n, K, d_z/d_x ))$ is a $K$-dimensional subspace of the complex Hilbert space
$\C^{2^{n}}$. The code corrects all qubit-flip errors up to $\frac{d_x-1}{2}$
and all phase-shift errors up to $\frac{d_z-1}{2}$. An $((n, 2^k, d_z/d_x ))$
code is denoted by $[[n, k, d_z/d_x ]]$.
\end{df}
	
We here construct families of AQSCs derived from non-orientable surfaces by
applying similar approaches utilized in \cite{Albuquerque2022} and \cite{Silva2021}.

\begin{teo}\label{teo1}
If $S_g$ is a non-orientable surface of even genus $g = 2h$ then an AQSC
generated by a $\{p, q\}$ tessellation has the same parameters of an AQSC
generated by the same $\{p, q\}$ tessellation over an orientable surface $S_h$ of genus $h$.
\end{teo}
	
\demo 	
The fundamental region of the orientable surface $S_h$ is the polygon $\{4h, 4h\}$, and
the fundamental region of a non-orientable surface $S_g$ is the polygon
$\{2g, 2g\}$. Therefore, $h=2g$; hence, the fundamental region of both surfaces are
the same $P^{\prime} = \{2g, 2g\}$ polygon.

The code length $n$ is computed by $n=n_f\frac{p}{2}$.
But $n_f = \frac{\mu(P^{\prime})}{\mu(P)}$, where $P = \{p, q\}$. However,
since $p$ and $n_f$ are equal for these two cases, then $n$ is also equal for
the corresponding two codes.
	
In both cases -- orientable and non-orientable surfaces--, $k$ can be computed by
$k=2- \mathcal{X}$, where $\mathcal{X}$ is the Euler characteristic
of the surface. For an orientable surface $S_h$, we have $\mathcal{X}(S_h)=2-2h$,
and for a non-orientable surface $S_g$, one has $\mathcal{X}(S_g)=2-g$.
Then $k = g = 2h$ are equal for these two codes.
	
Finally, let us compute the corresponding distances. The computations of $d_x$ and $d_z$,
as shown in \ref{dist}, depend only on three magnitudes: $d_h$ (hyperbolic
distance between opposite edge-pairings of the fundamental region),
$l(p, q)$ (the edge length of the $\{p,q\}$ polygon) and $l(q, p)$ (the edge length of
the $\{q, p\}$ polygon). Since $d_x = \frac{d_h}{l(p,q)}$ and $d_z = \frac{d_h}{l(q,p)}$,
and both -- fundamental polygon and tessellation polygon-- are the same for these
two codes, it follows that the distances are also equal.
This concludes the proof.\fimh
	
In Theorem~\ref{teo1}, we consider non-orientable surfaces of even genus. However,
when considering a surface of an arbitrary genus we have the following result.

\begin{teo}\label{mainteo}
Let $S$ be a non-orientable surface of genus $g$. Then an AQSC
generated by a $\{p, q\}$ tessellation over $S$ has ratio $k/n$
better than the ratio of an AQSC generated by the same $\{p, q\}$ tessellation over an
orientable surface of the same genus $g$.
\end{teo}

\begin{proof}
Let $S_1$ be an orientable surface and $S_2$ be a non-orientable surface,
both with the same genus $g$. We want to prove that the ratio
$r_1=k_1/n_1$ of the AQSC $C_1$ derived from the $\{p, q\}$ tessellation over $S_1$
is lower than the ratio $r_2=k_2/n_2$ of the AQSC $C_2$ derived from the
$\{p, q\}$ tessellation over $S_2$.

For $C_1$ we have: $k_1=2g$ and $n_1=\frac{2pq(g-1)}{pq-2p-2q}$; thus,
\begin{equation}
	r_1 = \frac{g(pq-2p-2q)}{pq(g-1)}.
\end{equation}
On the other hand, for $C_2$, it follows that $k_2=g$ and $n_2 =\frac{pq(g-2)}{pq-2p-2q}$,
which implies
\begin{equation}
	r_2 = \frac{g(pq-2p-2q)}{pq(g-2)}.
\end{equation}

Therefore,
\begin{equation}
	q =\frac{r_1}{r_2}= \frac{\frac{g(pq-2p-2q)}{pq(g-1)}}{\frac{g(pq-2p-2q)}{pq(g-2)}} =
\frac{g-2}{g-1} < 1, \forall g \geq 2,
\end{equation}
and the result follows.

\end{proof}

In the sequence, we exhibit the parameters of some new AQSCs obtained from non-orientable
surfaces of odd genus.
\begin{center}
	\begin{longtable}[c]{|c|c|c|c|c|}
		\caption{Parameters of AQSCs from non-orientable surface of genus $g=5$, $d_h \approx 3.5796$}\\
		\hline
		$\{p,q\}$ & $n_f$ & $ l(p,q) $ & $ [[n,k,d_z/d_x]] $ \\
		\hline
		\endfirsthead
		\multicolumn{4}{c}%
		{\tablename\ \thetable\ -- \textit{Continued from previous page}} \\
		\hline
		$\{p,q\}$ & $n_f$ & $ l(p,q) $ & $ [[n,k,d_z/d_x]] $  \\
		\hline
		\endhead
		\hline \multicolumn{5}{r}{\textit{Continued on next page}} \\
		\endfoot
		\hline
		\endlastfoot
		\hline\hline
		\hline
		$\{p,q\}$ & $n_f$ & $ l(p,q) $ & $ [[n,k,d_z/d_x]] $ \\ [0.5ex]
		\hline\hline
		$ \{3,7\} $ &     $ 42 $ &   $ 1.0905 $ & \\
			$ \{7,3\} $ &  $ 18 $ &    $ 0.5663 $  & $ [[63,5,7/4]] $ \\
			\hline\hline
			$ \{3,8\} $ &   $ 24 $ &    $ 1.5286 $ & \\
			$ \{8,3\} $ &   $ 9 $  &  $ 0.7270 $ & $ [[36,5,5/3]] $\\
			\hline\hline
			$ \{3,9\} $ &  $ 18 $ & $ 1.8551 $  & \\
			$ \{9,3\} $ &     $ 6 $ &   $ 0.8192 $   & $ [[27,5,5/2]] $\\
			\hline\hline
			$ \{3,12\} $ &   $ 12 $ &   $ 2.5534 $  &  \\
			$ \{12,3\} $ &    $ 3 $ &   $ 0.9516 $  & $ [[18,5,4/2]] $ \\
			\hline\hline
			$ \{3,15\} $ &   $ 10 $ &   $ 3.0486 $ & \\
			$ \{15,3\} $ &   $ 2 $ &   $ 1.0070 $ & $ [[15,5,4/2]] $\\
			\hline\hline
			$ \{4,5\} $ &   $ 15 $ &   $ 1.2537 $  & \\
			$ \{5,4\} $ &  $ 12 $  &  $ 1.0613 $   & $ [[30,5,4/3]] $ \\
			\hline\hline
			$ \{4,7\} $ &   $ 7 $  &  $ 2.1408 $  & \\
			$ \{7,4\} $ &   $ 4 $  &  $ 1.4491 $  & $ [[14,5,3/2]] $ \\
			\hline\hline
			$ \{4,8\} $  &  $ 6 $ &   $ 2.4485 $  & \\
			$ \{8,4\} $ &    $ 3 $ &   $ 1.5286 $  & $ [[12,5,3/2]] $\\
			\hline\hline
			$ \{4,10\} $ &    $ 5 $ &    $ 2.9387 $  & \\
			$ \{10,4\} $ &    $ 2 $ &   $ 1.6169 $  & $ [[10,5,3/2]] $\\
		\hline	
		
	\end{longtable}
\end{center}


\begin{center}
	\begin{longtable}[c]{|c|c|c|c|c|}
		\caption{Parameters of AQSCs from non-orientable surfaces of genus $g=7$, $d_h \approx 4.3144$}\\
		\hline
		$\{p,q\}$ & $n_f$ & $ l(p,q) $ & $ [[n,k,d_z/d_x]] $ \\
		\hline
		\endfirsthead
		\multicolumn{4}{c}%
		{\tablename\ \thetable\ -- \textit{Continued from previous page}} \\
		\hline
		$\{p,q\}$ & $n_f$ & $ l(p,q) $ & $ [[n,k,d_z/d_x]] $  \\
		\hline
		\endhead
		\hline \multicolumn{5}{r}{\textit{Continued on next page}} \\
		\endfoot
		\hline
		\endlastfoot
		\hline\hline
		\hline
		$\{p,q\}$ & $n_f$ & $ l(p,q) $ & $ [[n,k,d_z/d_x]] $ \\ [0.5ex]
		\hline\hline
  $ \{3,7\} $ &   $ 70 $ &   $ 1.0905 $  & \\
			$ \{7,3\} $ &  $ 30 $ &   $ 0.5663 $  & $ [[105,7,8/4]] $ \\
			\hline\hline
			$ \{3,8\} $ &   $ 40 $ &   $ 1.5286 $  & \\
			$ \{8,3\} $ &  $ 15 $ &   $ 0.7270 $  & $ [[60,7,6/3]] $ \\
			\hline\hline
			$ \{3,9\} $ &  $ 30 $ &   $ 1.8551 $  & \\
			$ \{9,3\} $ &   $ 10 $ &   $ 0.8192 $  & $ [[45,7,6/3]] $ \\
			\hline\hline
			$ \{3,11\} $ &  $ 22 $ &    $ 2.3517 $  & \\
			$ \{11,3\} $ &   $ 6 $ &   $ 0.9210 $ & $ [[33,7,5/2]] $\\
			\hline\hline
			$ \{3,12\} $ &  $ 20 $ &   $ 2.5534 $  & \\
			$ \{12,3\} $ &   $ 5 $ &   $ 0.9516 $  & $ [[30,7,5/2]] $\\
			\hline\hline
			$ \{3,16\} $ &  $ 16 $ &   $ 3.1877 $  & \\
			$ \{16,3\} $ &   $ 3 $  &  $ 1.0186 $  & $ [[24,7,5/2]] $\\
			\hline\hline
			$ \{3,21\} $ &  $ 14 $  &  $ 3.7611 $  & \\
			$ \{21,3\} $ &   $ 2 $  &  $ 1.0529 $  & $ [[21,7,4/2]] $ \\
			\hline\hline
			$ \{4,5\} $ &  $ 25 $ &   $ 1.2537 $   & \\
			$ \{5,4\} $ &  $ 20 $ &   $ 1.0613 $  & $ [[50,7,5/4]] $\\
			\hline\hline
			$ \{4,6\} $ &  $ 15 $ &   $ 1.7627 $  & \\
			$ \{6,4\} $ &  $ 10 $ &   $ 1.3170 $  & $ [[30,7,4/3]] $ \\
			\hline\hline
			$ \{4,8\} $ &  $ 10 $ &   $ 2.4485 $  & \\
			$ \{8,4\} $ &   $ 5 $ &   $ 1.5286 $  & $ [[20,7,3/2]] $\\
			\hline\hline
			$ \{4,9\} $ &   $ 9 $ &   $ 2.7101 $ &   \\
			$ \{9,4\} $ &   $ 4 $ &   $ 1.5807 $ & $ [[18,7,3/2]] $ \\
			\hline\hline
			$ \{4,14\} $ &   $ 7 $ &   $ 3.6472 $ & \\
			$ \{14,4\} $ &   $ 2 $ &   $ 1.6900 $  & $ [[14,7,3/2]] $\\
             \hline	
			\end{longtable}
              \end{center}

\begin{center}
	\begin{longtable}[c]{|c|c|c|c|c|}
		\caption{Parameters of AQSCs from non-orientable surfaces with $g=9$, $d_h \approx 4.8414$}\\
		\hline
		$\{p,q\}$ & $n_f$ & $ l(p,q) $ & $ [[n,k,d_z/d_x]] $ \\
		\hline
		\endfirsthead
		\multicolumn{4}{c}%
		{\tablename\ \thetable\ -- \textit{Continued from previous page}} \\
		\hline
		$\{p,q\}$ & $n_f$ & $ l(p,q) $ & $ [[n,k,d_z/d_x]] $  \\
		\hline
		\endhead
		\hline \multicolumn{5}{r}{\textit{Continued on next page}} \\
		\endfoot
		\hline
		\endlastfoot
		\hline\hline
		\hline
		$\{p,q\}$ & $n_f$ & $ l(p,q) $ & $ [[n,k,d_z/d_x]] $ \\ [0.5ex]
		\hline\hline
$\{p,q\}$ & $n_f$ & $ l(p,q) $ & $ [[n,k,d_z/d_x]] $   \\ \hline \hline
			$ \{3,7\} $ &   $ 98 $ &   $ 1.0905 $&\\
			$ \{7,3\} $ &   $ 42 $ &   $ 0.5663 $ &  $ [[147,9,9/5]] $ \\
			\hline\hline
			$ \{3,8\} $ &  $ 56 $ &   $ 1.5286 $  &\\
			$ \{8,3\} $ &   $ 21 $ &    $ 0.7270 $ & $ [[84,9,7/4]] $\\
			\hline\hline
			$ \{3,9\} $ &   $ 42 $ &   $ 1.8551 $  & \\
			$ \{9,3\} $ &  $ 14 $  &  $ 0.8192 $  & $ [[63,9,6/3]] $\\
			\hline\hline
			$ \{3,12\} $ &   $ 28 $ &   $ 2.5534 $  & \\
			$ \{12,3\} $ &    $ 7 $ &    $ 0.9516 $  & $ [[42,9,6/2]] $ \\
			\hline\hline
			$ \{3,13\} $ &  $ 26 $ &   $ 2.7341 $  & \\
			$ \{13,3\} $ &    $ 6 $ &   $ 0.9748 $  & $ [[39,9,5/2]] $\\
			\hline\hline
			$ \{3,20\} $ &   $ 20 $ &   $ 3.6594 $  & \\
			$ \{20,3\} $ &    $ 3 $ &   $ 1.0481 $  & $ [[30,9,5/2]] $\\
			\hline\hline
			$ \{3,27\} $ &  $ 18 $ &   $ 4.2792 $  & \\
			$ \{27,3\} $ &    $ 2 $&    $ 1.0712 $ & $ [27,9,5/2]] $\\
			\hline\hline
			$ \{4,5\} $ &  $ 35 $ &   $ 1.2537 $  & \\
			$ \{5,4\} $ &   $ 28 $ &   $ 1.0613 $  & $ [[70,9,5/4]] $\\
			\hline\hline
			$ \{4,6\} $ &  $ 21 $ &   $ 1.7627 $  & \\
			$ \{6,4\} $ &  $ 14 $ &   $ 1.3170 $  & $ [[42,9,4/3]] $ \\
			\hline\hline
			$ \{4,8\} $ &   $ 14 $ &   $ 2.4485 $  & \\
			$ \{8,4\} $ &   $ 7 $  &  $ 1.5286 $  & $ [[28,9,4/2]] $\\
			\hline\hline
			$ \{4,11\} $ &   $ 11 $ &   $ 3.1422 $  & \\
			$ \{11,4\} $ &   $ 4 $ &   $ 1.6432 $  & $ [[22,9,3/2]] $\\
			\hline\hline
			$ \{4,18\} $ &    $ 9 $ &   $ 4.1637 $ & \\
			$ \{18,4\} $ &    $ 2 $ &   $ 1.7191 $  & $ [[18,9,3/2]] $\\
			\hline\hline
			$ \{5,8\} $ &    $ 8 $ &    $ 2.7609 $  & \\
			$ \{8,5\} $ &   $ 5 $ &   $ 2.0481 $  &$ [[20,9,3/2]] $\\
			\hline\hline
			$ \{5,15\} $ &    $ 6 $ &   $ 4.0698 $  & \\
			$ \{15,5\} $ &   $ 2 $ &   $ 2.1934 $  & $ [[15,9,3/2]] $\\
			\hline
\end{longtable}
\end{center}

\begin{center}
	\begin{longtable}[c]{|c|c|c|c|c|}
		\caption{Parameters of AQSCs from non-orientable surfaces with $g=11$, $d_h \approx 5.2548$}\\
		\hline
		$\{p,q\}$ & $n_f$ & $ l(p,q) $ & $ [[n,k,d_z/d_x]] $ \\
		\hline
		\endfirsthead
		\multicolumn{4}{c}%
		{\tablename\ \thetable\ -- \textit{Continued from previous page}} \\
		\hline
		$\{p,q\}$ & $n_f$ & $ l(p,q) $ & $ [[n,k,d_z/d_x]] $  \\
		\hline
		\endhead
		\hline \multicolumn{5}{r}{\textit{Continued on next page}} \\
		\endfoot
		\hline
		\endlastfoot
		\hline\hline
		\hline
		$\{p,q\}$ & $n_f$ & $ l(p,q) $ & $ [[n,k,d_z/d_x]] $ \\ [0.5ex]
		\hline\hline
$\{p,q\}$ & $n_f$ & $ l(p,q) $ & $ [[n,k,d_z/d_x]] $   \\ \hline \hline
$\{p,q\}$ & $n_f$ & $ l(p,q) $ & $ [[n,k,d_z/d_x]] $   \\ \hline \hline
			$ \{3,7\} $ & $ 126 $ &   $ 1.0905 $  & \\
			$ \{7,3\} $ &  $ 54 $ &   $ 0.5663 $  &$ [[189,11,10/5]] $\\
			\hline\hline
			$ \{3,8\} $ &  $ 72 $ &   $ 1.5286 $  & \\
			$ \{8,3\} $ &  $ 27 $ &   $ 0.7270 $  & $ [[108,11,8/4]] $ \\
			\hline\hline
			$ \{3,9\} $ &  $ 54 $ &   $ 1.8551 $  & \\
			$ \{9,3\} $ &  $ 18 $ &   $ 0.8192 $  & $ [[81,11,7/3]] $ \\
			\hline\hline
			$ \{3,12\} $ &  $ 36 $ &   $ 2.5534 $  & \\
			$ \{12,3\} $ &    $ 9 $ &   $ 0.9516 $  & $ [[54,11,6/3]] $ \\
			\hline\hline
			$ \{3,15\} $ &   $ 30 $ &    $ 3.0486 $  & \\
			$ \{15,3\} $ &    $ 6 $ &   $ 1.0070 $  & $ [[45,11,6/2]] $ \\
			\hline\hline
			$ \{3,24\} $ &   $ 24 $ &   $ 4.0374 $  & \\
			$ \{24,3\} $ &    $ 3 $ &   $ 1.0638 $  & $ [[36,11,5/2]] $ \\
			\hline\hline
			$ \{3,33\} $ &  $ 22 $ &   $ 4.6883 $  & \\
			$ \{33,3\} $ &   $ 2 $ &   $ 1.0803 $  & $ [[33,11,5/2]] $ \\
			\hline\hline
			$ \{4,6\} $ &   $ 27 $ &   $ 1.7627 $  & \\
			$ \{6,4\} $ &  $ 18 $  &  $ 1.3170 $  & $ [[54,11,4/3]] $ \\
			\hline\hline
			$ \{4,7\} $ &  $ 21 $  &  $ 2.1408 $  & \\
			$ \{7,4\} $  &  $ 12 $ &   $ 1.4491 $  & $ [[42,11,4/3]] $ \\
			\hline\hline
			$ \{4,8\} $ &   $ 18 $ &    $ 2.4485 $  & \\
			$ \{8,4\} $ &    $ 9 $ &   $ 1.5286 $  & $ [[36,11,4/3]] $ \\
			\hline\hline
			$ \{4,10\} $ &   $ 15 $ &   $ 2.9387 $   & \\
			$ \{10,4\} $ &   $ 6 $ &   $ 1.6169 $  & $ [[30,11,4/2]] $ \\
			\hline\hline
			$ \{4,13\} $ &  $ 13 $ &   $ 3.4932 $   & \\
			$ \{13,4\} $ &   $ 4 $ &   $ 1.6780 $   & $ [[26,11,4/2]] $ \\
			\hline\hline
			$ \{4,16\} $ &  $ 12 $ &   $ 3.9225 $  & \\
			$ \{16,4\} $ &   $ 3 $ &   $ 1.7073 $  & $ [[24,11,4/2]] $ \\
			\hline\hline
			$ \{4,22\} $ &   $ 11 $ &   $ 4.5720 $ & \\
			$ \{22,4\} $ &   $ 2 $  &  $ 1.7337 $  & $ [[22,11,4/2]] $ \\
			\hline\hline
			$ \{6,12\} $ &    $ 6 $ &    $ 3.7556 $  & \\
			$ \{12,6\} $ &   $ 3 $ &   $ 2.5534 $  & $ [[18,11,3/2]] $\\
			\hline
\end{longtable}
\end{center}

\begin{table}[H]
\centering
\label{families}
\begin{tabular}{|l|l|l|}
			\hline
			$\{p,q\}$ & $n_f$ & $ [[n,k,d_z/d_x]] $   \\ \hline \hline
			$ \{7,3\} $ & $ 6(g-2) $ &   $ [[21(g-2), g, d_z/d_x]] $ \\ \hline			
			$ \{8,3\} $ & $ 3(g-2) $ &   $ [[12(g-2), g, d_z/d_x]] $ \\ \hline
			$ \{9,3\} $ & $ 2(g-2) $ &   $ [[9(g-2), g, d_z/d_x]] $ \\ \hline
			$ \{12,3\} $ & $ (g-2) $ &   $ [[6(g-2), g, d_z/d_x]] $ \\ \hline
			$ \{5,4\} $ & $ 4(g-2) $ &   $ [[10(g-2), g, d_z/d_x]] $ \\ \hline
			$ \{6,4\} $ & $ 2(g-2) $ &   $ [[6(g-2), g, d_z/d_x]] $ \\ \hline
			$ \{8,4\} $ & $ (g-2) $ &   $ [[4(g-2), g, d_z/d_x]] $ \\ \hline
			$ \{7,3\} $ & $ 6(g-2) $ &   $ [[21(g-2), g, d_z/d_x]] $ \\ \hline
		
\end{tabular}
\caption{Parameters of AQSCs from $2g$ polygon}
\end{table}

\begin{figure}[H]
	\centering
	\includegraphics{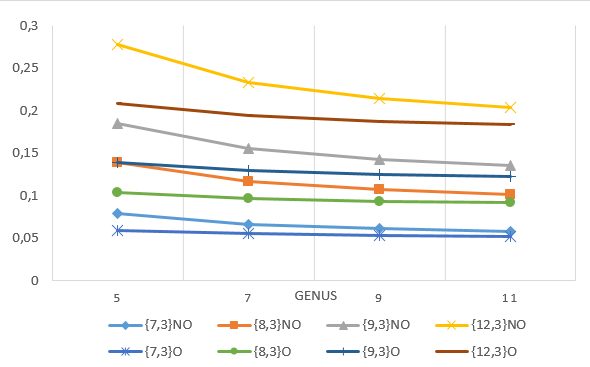}
	\caption{Comparison between $k/n$ for $\{p, q\}$ tessellation on orientable surface ($\{p,q\}$O) and over a non orientable surface ($\{p,q\}$NO) for each odd genus $g$.}
	\label{ratio}
\end{figure}

\begin{figure}[H]
	\centering
	\includegraphics{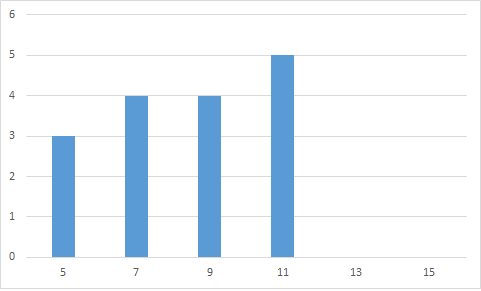}
	\caption{$d_z - d_x$ versus $g$ from $\{7,3\}$ tessellation}
	\label{difference}
\end{figure}

	
\section{Final remarks}\label{final}
		
We have constructed asymmetric quantum surface codes over non-orientable surfaces.
The new AQSC obtained from a $\{p, q\}$
tessellation over a non-orientable surface of even genus $g = 2h$ has the same parameters as
the AQSC generated by a $\{p, q\}$ tessellation over an orientable
surface $S_h$ of genus $h$. Moreover,
if $S$ is a non-orientable surface of genus $g$, we have shown that the new
AQSC constructed on a $\{p, q\}$ tessellation over $S$ has
ratio $k/n$ better than the ratio of an AQSC constructed on the same tessellation
over an orientable surface of the same genus $g$. In some specific cases, the new
asymmetric codes present great asymmetry
between $d_z$ and $d_x$.
For instance, if we consider the $\{7, 3\}$ tessellation, we can
observe that when considering surfaces of odd genus, the difference
between $d_z$ and $d_x$ increases.

\bibliography{AsymmetricnonorientF}

\vspace{2ex}

Waldir Silva Soares Jr.\\
Department of Mathematics, Campus de Pato Branco, UTFPR, Universidade Técnica Federal do Parana \\
Via do Conhecimento, s/n-KM 01-Fraron, Pato Branco - PR, 85503-390, Brazil\\
Email: waldirjunior@utfpr.edu.br \\

Douglas Fernando Copatti\\
Department of Mathematics, Instituto Federal do Paraná -  Campus Pitanga \\ 
Rua José de Alencar, 1.080 ? Vila Planalto, Pitanga - PR, 85.200-000, Brazil\\
Email: douglascopatti@gmail.com \\

Giuliano Gadioli La Guardia\\
Department of Mathematics and Statistics, State University of Ponta Grossa (UEPG), 
Campus Uvaranas - Bloco L - Uvaranas Av. General Carlos Cavalcanti, 4748\\
84030-900, Ponta Grossa - PR, Brazil\\
E-mail: gguardia@uepg.br \\

Eduardo Brandani da Silva \\
Department of Mathematics, Maringa State University - UEM \\
Av. Colombo 5790, Maringa - PR, 87020-900 - Brazil \\
Email: ebsilva@uem.br\\

\end{document}